%% file: main.tex
\documentclass[lettersize,journal]{IEEEtran}

\usepackage{easyReview}
\usepackage{amsmath, amsfonts, bbm}
\usepackage{multirow}  
\usepackage{booktabs}
\usepackage[ruled,vlined]{algorithm2e}

\usepackage{caption}
\usepackage{array}
\usepackage{textcomp}
\usepackage{stfloats}
\usepackage{url}
\usepackage{verbatim}
\usepackage{graphicx}
\usepackage{diagbox, slashbox}
\usepackage{subcaption}
\usepackage{cite}
\usepackage{amsthm}
\usepackage{hyperref}
\usepackage{easyReview}

\newtheorem{defi}{Definition}
\newtheorem{thm}{Theorem}
\newtheorem{lem}{Lemma}

\newtheorem{rmk}{Remark}

\newcommand{\x}{\boldsymbol{x}}
\newcommand{\y}{\boldsymbol{y}}
\newcommand{\z}{\boldsymbol{z}}

\newcommand{\bv}{\boldsymbol{v}}
\newcommand{\rl}{\boldsymbol{\ell}}
\newcommand{\hatl}{\boldsymbol{\hat{\ell}}}
\newcommand{\hatt}{\hat{\theta}}

\newcommand{\ar}{\mathcal{A}_\mathrm{R}}
\newcommand{\aj}{\mathcal{A}_\mathrm{J}}
\newcommand{\dar}{\Delta_{|\mathcal{A}_\mathrm{R}|}}
\newcommand{\mr}{\mathrm{R}}
\newcommand{\mj}{\mathrm{J}}
\newcommand{\U}{\textbf{U}}
\newcommand{\mH}{\mathcal{H}}

\definecolor{orange}{RGB}{255,107,0}

\begin{document}

\title{Learning an Opponent-aware Anti-jamming Strategy via Online Convex Optimization}

\author{
\IEEEauthorblockN{Liangqi Liu, Wenqiang Pu,~\IEEEmembership{Member,~IEEE}, Yingru Li, Zhi-Quan Luo,~\IEEEmembership{Fellow,~IEEE}}


\thanks{Liangqi Liu, Wenqiang Pu, Yingru Li, and Zhi-Quan Luo are with Shenzhen Research Institute of Big Data, The Chinese University of Hong Kong-Shenzhen, Guangdong,
China. Part of this work is presented at 2024 IEEE 13rd Sensor Array and Multichannel Signal Processing Workshop~\cite{liu2024radar}, this journal version provides detailed models, complete theoretical analysis, and additional results.}

}



\maketitle

\begin{abstract}
The dynamic competition against intelligent jammer systems presents a significant challenge to modern radar. Traditional active anti-jamming strategy learning methods often suffer from low sample efficiency and fail to fully exploit the structures of the adversary jammer. To reveal the inherent structure, this paper adopts an Online Convex Optimization (OCO) framework to capture the competition between a frequency agile radar and a digital radio frequency memory (DRFM)-based intelligent jammer. Recognizing that conventional OCO algorithms also suffer from suboptimal sample efficiency, two refined algorithms are developed that incorporate unbiased gradient estimators specifically tailored to the unique characteristics of DRFM-based jammers. Our theoretical analysis of the regret bound indicates significant improvements in long-term performance compared to standard OCO. The simulation results consistently show that our algorithms outperform traditional OCO and reinforcement learning baselines, achieving faster convergence and better anti-jamming performance.



\end{abstract}


\section{Introduction}

\input{sections/intro}

\section{System Model}\label{sec:system-model}
\input{sections/signal}

\section{OCO-based Radar Strategy Design}\label{sec:oco-base}
\input{sections/oco}

\section{Proposed Algorithms}\label{sec:alg}
\input{sections/algo}


\section{Experiments}\label{sec:expe}
\input{sections/expe}

\section{Conclusion}\label{sec:conclu}
In this work, the competition between a subpulse-level frequency-agile radar and DRFM-based main-lobe intelligent jammer is modeled as an online convex optimization (OCO) problem. Within this OCO framework, useful information about various jamming strategies is naturally embedded in the gradients of the cost function. By careful modeling intelligent jammers, we develop two novel algorithms that outperform conventional OCO benchmarks. Sub-linear static and universal regret bounds are provided, and numerical simulations demonstrate a significant enhancement in sample efficiency.

\section*{Appendix}
\input{sections/appendix}


\bibliographystyle{IEEEtran}
\bibliography{ref}

\end{document}

%% file: sections/intro.tex
\IEEEPARstart{M}{odern} electronic warfare (EW) presents increasingly sophisticated challenges for radar systems, particularly from intelligent jammers~\cite{adamy2001ew,de2018introduction}. Among these, main lobe jamming represents a critical threat, wherein jammers strategically position themselves within the radar's main beam. While traditional passive signal-processing methods~\cite{spezio2002electronic, maksimov1979radar,neri2001anti} have been developed to counter such jamming, their effectiveness remains limited by underlying assumptions that often fail in real-world EW scenarios. For instance, blind source separation techniques~\cite{geMainlobeJammingSuppression2018} become ineffective when the angular separation between radar and jammer decreases.


To address the limitations of passive methods, active strategy designs have been introduced as an alternative, enabling radar systems to dynamically adjust their transmission parameters. One prominent approach is exploiting the frequency agility~\cite{axelssonAnalysisRandomStep2007a, zhouruixueCoherentSignalProcessing2015, bicaGeneralizedMulticarrierRadar2016}, where the radar unpredictably varies its carrier frequency across a wide spectrum. By formulating the carrier frequency selection problem as a sequential decision-making process, reinforcement learning (RL) methods have gained significant attention. Notable works include a deep RL-based strategy design for active antagonism~\cite{liRadarActiveAntagonism2021}, the implementation of Deep Q-Network (DQN) with detection probability as a performance metric~\cite{liDeepQNetworkBased2019}, and the application of dueling double DQN  for airborne radar anti-jamming waveform design~\cite{zhengAirborneRadarAntiJamming2022}.

In addition to RL methods, game-theoretic approaches have also been extensively explored for designing anti-jamming strategies. For instance, two-person zero-sum games have been employed to model radar-jammer competition~\cite{songMIMORadarJammer2012}, where strategies are defined based on the information available to each player. Similarly, Stackelberg equilibrium approaches have been applied to radar-target interactions in cluttered environments, as shown in~\cite{lanMIMORadarTarget2015}, where hierarchical decision structures are utilized to model dominance between radar and jammer. Recently, extensive-form games with imperfect information have been investigated through neural fictitious self-play (NFSP)~\cite{liNeuralFictitiousSelfPlay2022} and counterfactual regret minimization~\cite{li2022counterfactual}, offering significant reductions in computational complexity for multi-stage decision-making scenarios. Beyond frequency agility, other practical aspects of anti-jamming have also been investigated in recent years. For example, authors in~\cite{gengRadarJammerIntelligent2023} explore scenarios where the jammer dynamically allocates power, deriving a Nash equilibrium (NE) using DQN-based algorithms. A Stackelberg game is constructed in~\cite{fengRadarJammingStrategy2023} by incorporating recognition time into the utility function, while authors in~\cite{ailiyaAdaptationFrequencyHopping2022} proposes an RL-based framework to adaptively adjust frequency hopping interval and pulse width.

Despite advancements in both RL-based and game-based studies, significant challenges remain. Many of these methods suffer from low sample efficiency, making them computationally expensive and impractical for real-time operations. Additionally, some works~\cite{zhangNewSchemeTarget2022, hanTwodimensionalAntijammingCommunication2017, ailiyaAdaptationFrequencyHopping2022} rely on unrealistic assumptions about the prior knowledge of jamming types or behaviors, limiting their applicability in dynamic and unpredictable electronic warfare scenarios. These limitations highlight a naive, yet inadequately addressed question: \textit{what is the essential structure that governs the learning of an effective anti-jamming strategy?} While a seemingly intuitive answer might attribute this to the intelligence of the jammer, such an explanation is overly abstract and lacks insights.

In this work, we revisit the anti-jamming strategy learning problem through the lens of Online Convex Optimization (OCO)~\cite{zinkevich2003online, shalev2012online, hazan2007logarithmic}, which offers a rigorous and elegant framework to study this specific problem. This perspective enables us to identify a critical structural insight: useful information about the jammer's behavior can be inherently encoded in the gradient of the cost function. By leveraging this insight, we demonstrate that online interaction samples can be systematically utilized to effectively learn the anti-jamming strategies. This structured approach not only provides theoretical clarity but also introduces a practical mechanism for incorporating prior knowledge of the jammer, which significantly enhances the real-world applicability of the learning process. The main contributions of this paper are summarized as follows:


\begin{itemize}
    \item \textbf{OCO framework}: The radar's anti-jamming strategy is formulated within an Online Convex Optimization (OCO) framework, embedding key characteristics of the jammer, such as history-dependent behavior, into the gradient estimation of the cost functions. This enables an adaptive and efficient online strategy that dynamically responds to jamming threats.

    \item \textbf{Two OCO algorithms}: For scenarios where no prior knowledge of the jammer's behavior is available, an action modeling-based algorithm is proposed, leveraging the jammer's action space to ensure a sublinear regret bound, independent of the jammer's strategy. For jammers influenced by interaction history, an opponent modeling-based algorithm is introduced, capturing the jammer's decision-making patterns to develop a more adaptive and efficient anti-jamming strategy.

    \item \textbf{Regret analysis}: Theoretical regret bounds are derived for both static and universal regret metrics. The action modeling-based algorithm achieves a static regret bound of $\Tilde{\mathcal{O}}(\sqrt{N})$, aligning with the OCO framework's lower bound, where $N$ is the total number of rounds. The opponent modeling-based algorithm achieves a sublinear universal regret bound of $\Tilde{\mathcal{O}}(\sqrt{N|\mH|})$, where $\mH$ represents the recent history space, significantly improving sample efficiency for history-dependent jamming strategies.
\end{itemize}

The remainder of this paper is organized as follows. Section~\ref{sec:system-model} introduces the system model and problem formulation, and Section~\ref{sec:oco-base} elaborates on the proposed OCO framework. Section~\ref{sec:alg} presents the two proposed algorithms, along with their theoretical regret analysis. Section~\ref{sec:expe} provides numerical experiments to validate the proposed methods and finally Section~\ref{sec:conclu} concludes this work.

\noindent\textbf{Notations:} We use script symbols such as \( \mathcal{A} \) to denote a finite set, with \( |\mathcal{A}| \) representing its cardinality. Boldface symbols like \( \x \) denote vectors, and for a vector \( \x \) defined over \( \mathcal{A} \), which has a dimension of \( |\mathcal{A}| \), we use \( \x(a) \) to represent the scalar value corresponding to the element \( a \in \mathcal{A} \). The \( k \)-dimensional probability simplex is represented by \( \Delta_k \) and the notation \( [N] \) refers to the set \( \{1, 2, \ldots, N\} \). Other notations are defined as they appear in the text.

%% file: sections/signal.tex
\subsection{Signal Models}\label{sec:sig-model}
Consider a subpulse-level frequency-agile (FA) radar system operating in monopulse transmission mode, as illustrated in Fig.~\ref{fig:sig-01}. In this system, each pulse contains $M$ subpulses, with each subpulse independently selecting its carrier frequency from a predefined set $\mathcal{F} = f_c+\{f_1, f_2, \ldots, f_L\}$, where $f_c$ is the central frequency. The frequencies within $\mathcal{F}$ are uniformly spaced with a constant step size $\Delta f$, such that each frequency $f_i$ for $i = 2, 3, \ldots, L$ is defined as $f_i=f_{i-1}+\Delta f$. The waveform of the $n$-th pulse at time $t$, denoted by $s_n(t)$, is formulated as:
\begin{equation}\label{equ:signal_radar}
s_n(t) = \sum_{m=0}^{M-1}\text{rect}\left(\frac{t-mT_c}{T_c}\right)a(t)\exp{(j2\pi f^{\mr}_{n,m} t)},
\end{equation}
where $T_c$ is the duration time of each subpulse, $f^{\mr}_{n,m}\in\mathcal{F}$ represents the sub-carrier frequency for subpulse $m$ of the $n$-th pulse, $a(t)$ is the complex envelope of baseband signal, and $\text{rect}(t)$ is the rectangle function defined by $\text{rect}(t)=1$ if $0\leq t\leq 1$ and $\text{rect}(t)=0$ otherwise.
\begin{figure}
    \centering
    \includegraphics[width=0.95\linewidth]{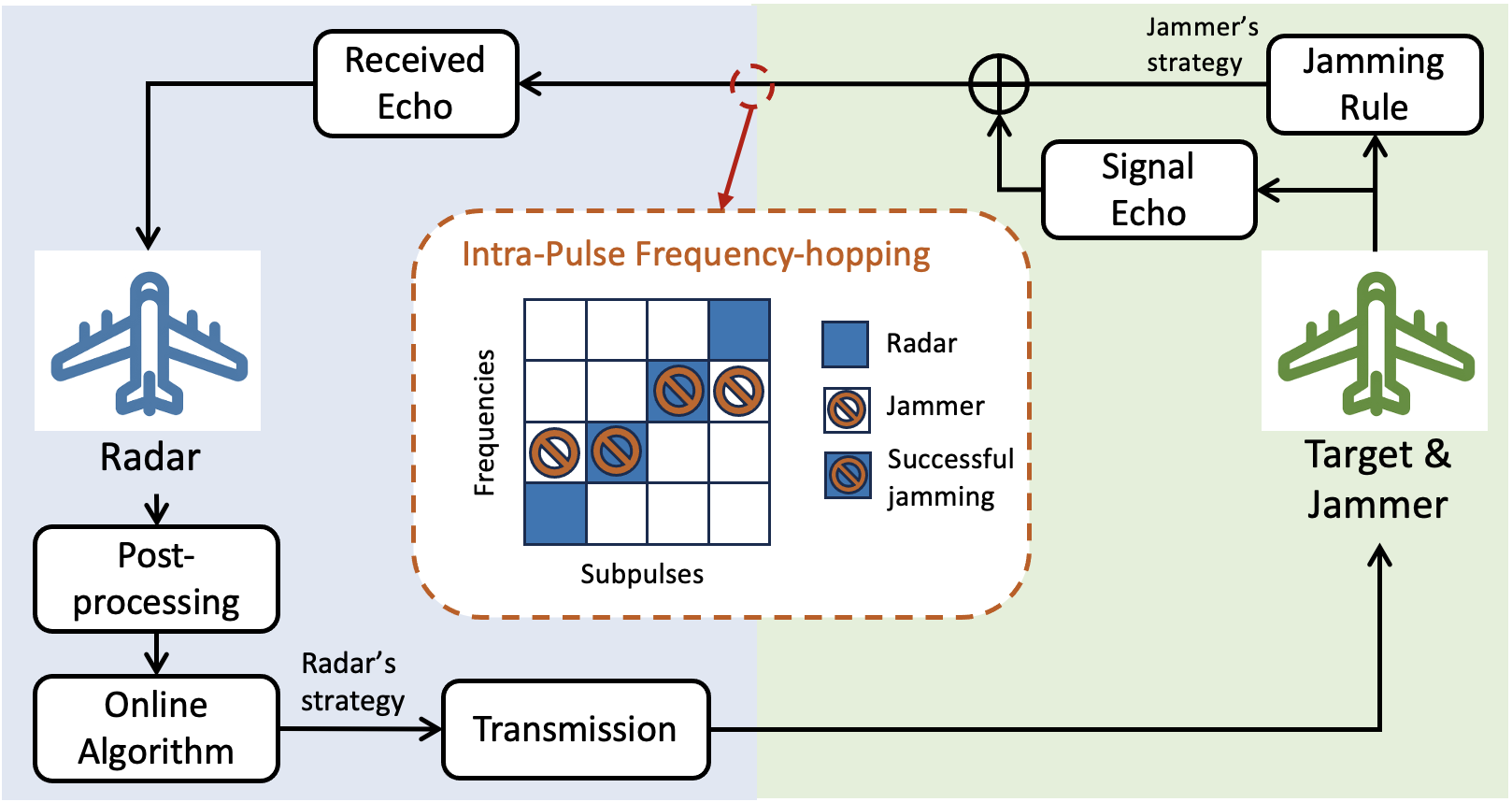}
    \caption{Illustration of the anti-jamming scenario. (The radar system implements sub-pulse frequency-hopping, while the jammer chooses a narrow noise jamming.)}
    \label{fig:sig-01}
\end{figure}

The subpulse-level frequency agile has been employed in various advanced radar systems, as discussed in recent literature~\cite{ailiyaAdaptationFrequencyHopping2022, liRadarActiveAntagonism2021}. Compared to traditional pulse-to-pulse frequency-agile radar systems~\cite{axelssonAnalysisRandomStep2007a}, this technique demonstrates superior capability to counteract complex jamming types by dynamically varying the sub-carrier frequencies within each pulse. Furthermore, intrapulse frequency hopping is gaining increased significance because modern jammer systems can detect a radar's carrier frequency within only tens of nanoseconds~\cite{nguyenRealtimeProtocolawareReactive2014}—a duration markedly shorter than the typical radar pulse width. The sub-pulse frequency hopping technique provides a critical advantage in electronic warfare environments.

Given the rapid response capabilities of the jammer system, the noise-modulated jamming signal is considered to be transmitted 'immediately' after receiving the radar pulse. Assume a distance $R$ between the radar and the target/jammer, the transmitted jamming signal is mathematically described as:
\begin{equation}\label{equ:signal_jammer}
j_n(t) =\sum_{m=0}^{M-1}\text{rect}\left(\frac{t-T_d - mT_c}{T_c} \right)b(t)\exp{(j2\pi f^\mj_{n,m} t)},
\end{equation}
where $T_d=R/c$ (with $c$ representing the speed of light) denotes the propagation time from the radar to the target/jammer, $b(t)$ represents the noise-modulated signal envelope, and $f^\mj_{n,m}$ indicates the carrier frequency for sub-pulse $m$ at pulse $n$.

Due to the subpulse-level frequency agility of the radar system, $f^\mj_{n,m}$ is considered to vary within pulse $n$ to ensure effective jamming. Each $f^\mj_{n,m}$ is selected based on jammer's employed strategy. Note that the jamming signal model in~\eqref{equ:signal_jammer} covers a range of practical jamming types, including traditional suppression jamming types such as Noise Frequency Modulation Jamming (NFMJ)~\cite{neng1995survey}, and novel jamming techniques based on DRFM devices, like Interrupted Sampling Repeater Jamming (ISRJ)~\cite{wang2007mathematic} and Smart Noise Jamming (SNJ)~\cite{zhang2022radar}.
    

At the radar receiver, the incoming signal consists of both the reflected radar pulse echo, as described in E.q.~\eqref{equ:signal_radar}, and the jamming signal from E.q.~\eqref{equ:signal_jammer}. The received signal $r_n(t)$ is given by
\begin{equation}\label{equ:signal_echo}
r_n(t) = \alpha s_n(t-T_d) + j_n(t - T_d) + w_n(t) 
\end{equation}
where $\alpha\in\mathbb{C}$ is the signal attenuation factor and $w_n(t)$ is the zero-mean Gaussian noise.

Fig.~\ref{fig:spectrogram} presents a spectrogram example for the frequency-hopping scenario described in Fig.~\ref{fig:sig-01}, visualizing both the radar signal $s_n(t)$ and the received  signal $r_n(t)$. In the spectrogram, the straight regions represent the radar's signals, whereas the irregular regions indicate jamming noise. The red ellipsoids highlight the two sub-pulses that have been successfully jammed. A Signal-to-Interference-plus-Noise Ratio (SINR) of $1.65$dB is calculated from the echo spectrum using the post-processing procedure described in the following section.
\begin{figure}
    \centering
    \includegraphics[width=0.99\linewidth]{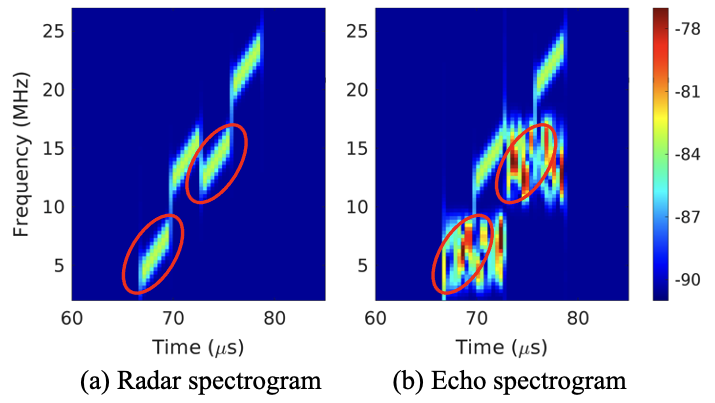}
    \caption{Spectrogram example of the transmit and echo signal. The subpulse frequency hopping follows the sequence $\{6,14,14,22\}$ MHz, while the jammer follows $\{6,6,14,14\}$ MHz. Detailed settings are presented in Table~\ref{tab:para}. }
    \label{fig:spectrogram}
\end{figure}

\subsection{Post-processing of the Received Signal}\label{sec:post-processing}
The received signal $r_n(t)$ requires post-processing to extract valuable information and identify the presence of any jamming signals and here we omit pulse index $n$ for presentation simplicity. This processing involves bandpass filtering to isolate the desired frequency band and matched filtering to enhance the signal of interest within the received signal, as illustrated in Fig.~\ref{fig:sig-02}. Processing details are provided in Appendix~\ref{appdix:postprocessing}.

After the post-processing, the frequency information $f^\mj_m$ can be extracted and the SINR which offers a direct measure of the quality for target detection is calculated for each subpulse, 
\begin{equation}\label{equ:sinr}
\text{SINR}(f^\mr_m, f^\mj_m\mid\theta) = \frac{P_\mr}{P_{0}+P_\mj\mathbbm{1}(f^\mr_m=f^\mj_m)},\ \forall m\in[M]
\end{equation}
where $\theta=[P_\mr, P_\mj, P_{0}]$ represents the received the radar signal pwoer $P_\mr$, jammer signal power $P_\mj$, and noise power $P_{0}$, respectively. The indicator function $\mathbbm{1}(f^\mr_m=f^\mj_m)$ determines whether the received signal in sub-pulse $m$ is jammed, i.e., when the frequencies of radar and jammer coincide.

Notice that $\theta$ is unknown in advance and can be estimated iteratively through post-processing. In Section~\ref{sec:oco-base}, we consider an ideal scenario where $\theta$ is assumed to be known and this provides a detailed illustration of the proposed OCO framework, with SINR serves as a component of the cost function. Subsequently, in Section~\ref{sec:alg}, the uncertainty of $\theta$ is addressed within the proposed algorithms, further enhancing the framework's applicability in practical scenarios.
\begin{figure}
    \centering
    \includegraphics[width=0.99\linewidth]{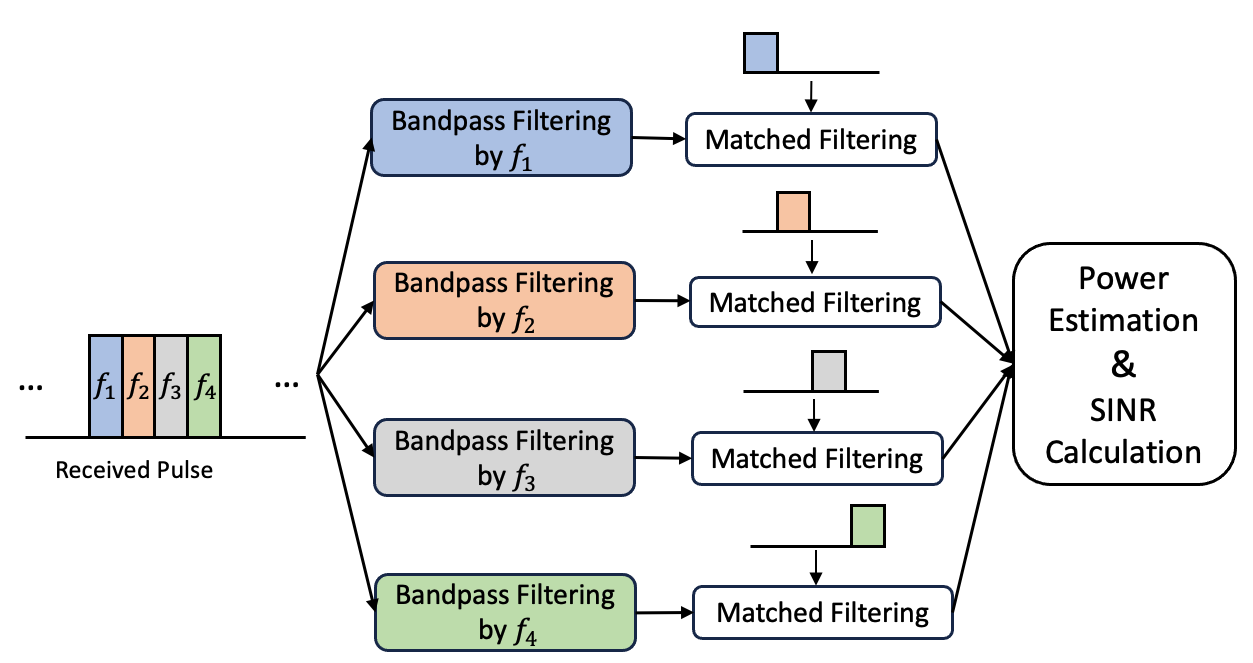}
    \caption{Post-processing for the received signal.}
    \label{fig:sig-02}
\end{figure}

%% file: sections/oco.tex
\subsection{Online Optimization Formulation}
The received signal model in E.q.~\eqref{equ:signal_echo} is an aggregation of both radar's and jammer's signals, which can be regarded as a one-round interaction between the radar and jammer. Protocol~\ref{alg:protocol} illustrates this iterative interaction, where the radar system must adaptively change its transmitted signal to improve the corresponding SINR of the received signal. To mathematically represent this online sequential decision-making procedure, we introduce an online convex optimization (OCO) framework in this section. Below, basic components of this framework for pulse $n$ are introduced.

{
    \renewcommand{\algorithmcfname}{Protocol} 
    \begin{algorithm}
    \SetAlgoCaptionSeparator{.}  
    \caption{Iterative anti-jamming procedure}
    \label{alg:protocol}
    \SetAlgoLined
    \For{$n \gets 1$ \KwTo $N$}{
        Radar transmits the signal $s_n(t)$\;
        Jammer detects $s_n(t)$ and transmits $j_n(t)$\;
        Radar receives the echo $r_n(t)$.
    }
    \end{algorithm}
}

\begin{itemize}
\item \textbf{Radar's Action Set} \(\ar\): The radar's action set is defined as \(\ar = \mathcal{F}^M\), where each element 
$$a_n = (f^\mr_{n,1}, f^\mr_{n,2}, \ldots, f^\mr_{n,M}) \in \ar$$ 
represents the radar's selection of \(M\) sub-carrier frequencies for the \(n\)-th pulse. This set encapsulates all possible frequency combinations available to the radar.

\item \textbf{Radar's Strategy} \(\x_n\): The radar's strategy, from which the action \(a_n\) is drawn, lies within the probability simplex \(\dar\). This ensures that \(\x_n\) represents a valid probabilistic distribution over the action set \(\ar\), allowing for randomized decision-making.

\item \textbf{Cost Function} \(f_n\): The cost function \(f_n\) is a convex, linear function defined as 
$$f_n(\x) = \langle \rl_n, \x \rangle : \dar \rightarrow \mathbb{R}.$$
Here, \(\rl_n \in \mathbb{R}^{|\ar|}\) denotes the cost vector associated with the action set \(\ar\), which quantifies the performance or penalty of each action in the given scenario.
\end{itemize}

Note that the cost vector $\rl_n$ is derived from the SINR value described in E.q.~\eqref{equ:sinr}, given as
\begin{equation}\label{equ:ln}
    \rl_n(a_n)=\frac{c-\overline{\text{SINR}}_n(a_n)}{c}\in[0,1], \forall a\in\ar,
\end{equation}
where action $a_n\in\ar$ represents a specific combination of sub-carrier frequencies. The term 
\begin{equation*}
\overline{\text{SINR}}_n(a_n)=\frac{1}{M}\sum_{m=1}^M \min\{\text{SINR}(f^\mr_{n,m},f^\mj_{n,m}\mid\theta), c\}
\end{equation*}
denotes the average SINR for action $a_n$, depending on the frequencies of jamming signal at pulse $n$. The constant $c$ is a SINR threshold that indicates a sufficiently high probability of detection.
\begin{rmk}[Limited feedback]\label{rmk:ln(an)}
Since the received signal $r_n(t)$ corresponds to radar's specific action $a_n$, only the value $\rl_n(a_n)$ can be evaluated according to E.q.~\eqref{equ:ln}. Other elements in $\rl_n$ cannot be directly inferred from $r_n(t)$, and we refer to such a single-value return as limited feedback.
\end{rmk}

The primary goal of the radar is to minimize the cost $f_n(a_n)$ in each round $n$, formally expressed as
\begin{equation}\label{equ:obj}
\begin{aligned}
\min_{\x_n} \quad & f_n(\x_n) \\
\text{s.t.} \quad & \x_n\in\dar.
\end{aligned}
\end{equation}

Determining the optimal strategy $\x^*_n$ for each round $n$ is challenging due to the dynamic and unpredictable nature of the jamming environment, where the function \(f_n(\cdot)\) is not known in advance. Consequently, in each round $n$, it becomes practical to utilize available historical information, such as $a_1, f_1(a_1),\ldots,a_{n-1}, f_{n-1}(a_{n-1})$, to determine $\x_n$. This sequential decision-making process is known as online optimization. After $N$ rounds, the performance of an online algorithm \(\mathcal{A}\) is retrospectively measured by the static regret (see Definition~\ref{def:static}), which compares the cumulated cost incurred by $\mathcal{A}$ to that of an optimal strategy
\begin{equation}\label{equ:x^*}
    \x^* = \arg\min_{\x\in\dar}\sum_{n=1}^N f_n(\x).
\end{equation}

\begin{defi}[Static regret]\label{def:static}
The static regret of an online algorithm \(\mathcal{A}\) is defined as 
\begin{equation}\label{equ:sta-regret}
     \text{S-Regret}_N(\mathcal{A})=\mathbb{E}\left[\sum_{n=1}^Nf_n(\x_n) - \sum_{n=1}^N f_n(\x^*)\right]
\end{equation}
where the expectation is taken over the randomness in the functions $f_n,\ n=1,2,\ldots,N$.
\end{defi}
Static regret is a widely used performance metric for an online algorithm, as it accounts for the non-stationary nature of the environments, i.e., the time-varying cost functions \(f_n(\cdot)\). The lower bound of the static regret is known to be \(\Omega(\sqrt{N})\) when $f_n$ is convex~\cite{hazanIntroductionOnlineConvex2023}.


\subsection{Online Mirror Descent}
One popular algorithm that achieves a regret bound \(\mathcal{O}(\sqrt{N})\) is the online gradient descent (OGD). This method iteratively updates the strategy by taking a gradient step,
\begin{equation*}
    \x_{n+1}=\text{Proj}_{\dar}[\x_n-\eta\nabla f_n(\x_n)]
\end{equation*}
where $\text{Proj}_{\dar}[\cdot]$ denotes the projection onto the feasible set $\dar$ and $\eta$ represents the learning rate. While this projection is typically performed with respect to the Euclidean distance, a more general approach is to use the Bregman divergence (Definition~\ref{def:bergman}), leading to what is known as online mirror descent (OMD)\cite{hazan2010extracting}. Compared with OGD, OMD also guarantees $\mathcal{O}(\sqrt{N})$ regret bound. Furthermore, using Bregman divergence for projection enhances efficiency and aligns closely with the geometry of the feasible set, thereby promoting stability and potentially achieving tighter regret bounds~\cite{hazanIntroductionOnlineConvex2023}. Below, we briefly review OMD method. 
\begin{defi}[Bregman divergence]\label{def:bergman}
Let $F: \mathbb{R}^{|\ar|}\rightarrow \mathbb{R}$ be a Legendre function~\cite{lattimoreBanditAlgorithms2020} defined on the domain $\mathcal{D}$. The Bregman divergence from the point $\x$ to $\boldsymbol{u}$ w.r.t. $F$ is defined as:
\begin{equation*}
    D_F(\x,\boldsymbol{u}) := F(\boldsymbol{u})-F(\x)-\langle\nabla F(\x), \boldsymbol{u}-\x\rangle
\end{equation*}
\end{defi}


Given a specific Legendre function $F$, the OMD update typically involves two steps:
the regularized gradient descent within the domain $\mathcal{D}$ of Bregman divergence $F(\cdot)$, and then mapping back to the primal domain $\dar$. Specifically,
\begin{equation}
\begin{cases}
\boldsymbol{u}_{n+1}&=\arg\min_{\x\in\mathcal{D}}\ \eta\langle\x,\nabla f_n(\x_n)\rangle+D_F(\x, \x_n), \\
\x_{n+1}&=\arg\min_{\x\in\dar}\ D_F(\x, \boldsymbol{u}_{n+1}).
\end{cases}
\end{equation}
One common choice of $F(\x)$ is the negative entropy function $F(\x)=\sum_a\x(a)\ln{\x(a)}$, which nicely aligns with the geometry of probability simplex, and thus the update rule of OMD becomes:
\begin{equation}\label{equ:omd}
\begin{cases}
\boldsymbol{u}_{n+1}(a)&=\x_{n}(a)\cdot e^{-\eta_n\rl_{n}(a)},\forall a,\\
\x_{n+1}&=\frac{\boldsymbol{u}_{n+1}}{\Vert\boldsymbol{u}_{n+1}\Vert_1},
\end{cases}
\end{equation}
where $\Vert\cdot\Vert_1$ denotes the $\ell_1$ norm operation and $\boldsymbol{u}_{n+1}(a)$ represents an unnormalized dual variable for the $a$-th component at iteration $n+1$. It can be observed that OMD in \eqref{equ:omd} requires gradient $\rl_n$ to be known, which is not applicable to our setting since only $\rl_n(a_n)$ is available to the radar (Remark~\ref{rmk:ln(an)}). In the following section, we introduce gradient estimators.


\subsection{Unbiased Gradient Estimator}
Although we cannot obtain $\rl_n$ explicitly, it is possible to construct an estimate $\hatl_n$ that satisfies $\mathbb{E}[\hatl_n]= \rl_n$. Lemma~\ref{lem:grad-est} shows the regret bound of an OCO algorithm when substituting $\hatl_n$ for $\rl_n$ of a linear function.
The proof for this lemma is presented in Appendix~\ref{proof:lem-grad}.
\begin{lem}[Regret bound with $\hatl_n$]\label{lem:grad-est}
Let $f_1,\ldots, f_N: \dar\rightarrow \mathbb{R}$ be a sequence of linear functions with gradients $\rl_1,\ldots,\rl_N$. For an OCO algorithm $\mathcal{A}$, define the updated points $\x_n$ as $\x_1\leftarrow\mathcal{A}(\emptyset),\x_n\leftarrow\mathcal{A}(\hatl_1,\ldots,\hatl_{n-1})$, where each $\hatl_n$
is an unbiased gradient estimator satisfying $\mathbb{E}[\hatl_n]=\rl_n$. Then the static regret defined in Definition~\ref{def:static} can be expressed in terms of the estimators $\hatl_n, \forall n\in[N]$ as
\begin{equation*}
\begin{aligned}
    \text{S-Regret}_N(\mathcal{A})= \mathbb{E}\left[ \sum_{n=1}^N \langle\hatl_n,\x_n\rangle - \sum_{n=1}^N \langle\hatl_n,\x^*\rangle \right]
\end{aligned}
\end{equation*}
where $\x^*$ is the optimal static strategy defined in E.q.~\eqref{equ:x^*}.
\end{lem}

Lemma~\ref{lem:grad-est} shows an equivalence when replacing the gradient with an unbiased estimator. However, if the estimator is biased, the lemma does not hold since a cumulative error term is introduced (c.f. Appendix~\ref{proof:lem-grad}). In online convex optimization literature, a traditional choice is the importance-weighted gradient estimator (IWE),
\begin{equation} \label{equ:iwe}
\hatl_{n}(a)=
\begin{cases}
\frac{\rl_{n}(a)}{\x_{n}(a)}, &a=a_n \\
0, & \text{Otherwise}.
\end{cases}
\end{equation}

It is easy to verify that IWE is an unbiased estimator. By replacing the gradient $\rl_n$ in E.q.~\eqref{equ:omd} with IWE, we obtain an algorithm called OMD-IWE. This algorithm is equivalent to the classical adversarial bandit algorithm known as Exp3~\cite{auerNonstochasticMultiarmedBandit2002a}. Using a constant learning rate $\eta=\sqrt{\log{|\ar|}/N|\ar|}$, OMD-IWE guarantees a regret bound of $\mathcal{O}(\sqrt{N|\ar|\log{|\ar|}})$.


%% file: sections/algo.tex
\subsection{Action Modeling}
The static regret of OME-IWE does not meet the lower bound $\Omega(\sqrt{N})$ mentioned in Definition.\ref{def:static}. The additional square-root dependency on $|\ar|$ arises due to the limited gradient estimation in each round. Note that the jammer's action space is a finite set, which provides useful information to enhance the OMD-IWE by developing an improved gradient estimator. Below, necessary notations for modeling jammer's action is introduced.
\begin{itemize}
    \item \textbf{Jammer's action set} $\mathcal{A}_\mathrm{J}$: Denote $b_n$ as the jammer action at round $n$ with $b_n\in\mathcal{A}_\mathrm{J}=\mathcal{F}^M\cup\mathcal{B}$, where $\mathcal{B}$ includes specific jamming actions such as the observe action adopted by interrupted sampling repeater jammer~\cite{wang2007mathematic}.
\item \textbf{Jammer's  strategy} $\y_n$: $b_n$ is sampled from $\y_n\in\Delta_{|\mathcal{A}_\mathrm{J}|}$, where $\Delta_{|\mathcal{A}_\mathrm{J}|}$ denotes the probability simplex.
\end{itemize}
With the help of notations $\mathcal{A}_\mathrm{J}$ and $\y_n$, function $f_n$ can be reformulated as the cost of a two-player matrix game.
\begin{defi}[Cost function]\label{def:phi}
Let $\mathbf{U}(\theta)\in\mathbb{R}^{|\mathcal{A}_R|\times|\mathcal{A}_J|}$ denote the cost matrix associated with each pair of actions $a\in\ar$ and $b\in\aj$. Specifically, each entry $\mathbf{U}[a,b\mid\theta]$ is defined as 
\begin{equation*}
    \mathbf{U}[a,b\mid\theta]=\frac{c-\overline{\text{SINR}}(a,b\mid\theta)}{c}\in[0,1],\forall a\in\ar,b\in\aj,
\end{equation*}
where 
$$\overline{\text{SINR}}(a,b\mid\theta)=\frac{1}{M}\sum_{m=1}^M\min\{\text{SINR}_m(f^\mr_m,f^\mj_m\mid\theta),c\},$$ and $\theta$ represents the scenario parameters defined in E.q.~\eqref{equ:sinr}. Then the function $\phi(\x,\y):\dar\times\Delta_{|\mathcal{A}_\mathrm{J}|}\rightarrow\mathbb{R}$ is defined to represent the cost related to both radar's and jammer's strategies:
\begin{equation}\label{equ:phi}
   \phi(\x_n, \y_n) = \x_n^T\U(\theta)\y_n
\end{equation}
The relationship between $\phi$ and $f_n$ is
\begin{equation*}
    \phi(\x_n, \y_n)=\mathbb{E}_{b_n\sim\y_n}[f_n(\x_n)].
\end{equation*}
\end{defi}
Notice that the jammer's action $b_n$ is implicitly implied by $f_n$ and for simplicity in notation, we keep the use of $f_n$. Based on the cost function $\phi$, the gradient $\rl_n$ with respect to $\x_n$ is given by
\begin{equation}\label{equ:grad-phi}
    \rl_n=\nabla_{\x_n} \phi(\x_n,\y_n)=\mathbf{U}(\theta)\y_n.
\end{equation}
\begin{rmk}[Estimation of $\theta$]\label{rmk:theta}
The scenario parameter $\theta$ can be estimated well within a few rounds by proper power spectrum estimation methods~\cite{bingham1967modern, boashash2015time, levin1965power}. Given that post-processing effectively isolates signal components, we consider that power spectrum estimation yields an unbiased estimate $\hatt$, i.e., $\mathbb{E}[\hatt]=\theta$. Consequently, it follows that $\mathbb{E}[\U[a,b]\mid \hatt]=\U[a,b]$. Furthermore, this work does not explore various power spectrum estimation techniques, as such an analysis is beyond its scope. Instead, we just assume that $\hatt$ is known with $\text{Var}\left[\U[a,b]\mid\hatt\right]\leq\sigma^2, \forall a,b,$ where $\sigma$ is a small constant.
\end{rmk}
Given the estimated scenario parameter $\hatt$, the cost matrix $\U(\hatt)$ provides a model-based mapping from any action pair $(a,b)$ to the corresponding normalized cost. Building on $\U(\hatt)$ and the jammer's action set $\aj$, an improved gradient estimator $\hatl_n$ is developed, enabling a more efficient update of anti-jamming strategies.
\begin{defi}[Action modeling based gradient estimator]\label{def:ame}
Let $\hatt$ be the estimated scenario parameters with $\mathbb{E}[\hatt]=\theta$, an unbiased gradient estimator $\hatl_n$ can be obtained as:
\begin{equation}\label{equ:ame}
    \hatl_n =\U{[:,b_n\mid\hatt]}
\end{equation}
where $\hatl_n$ is unbiased since $\mathbb{E}[\hatl_n]=\mathbb{E}[\U(\hatt)]\y_n=\rl_n.$
\end{defi}
The proposed \textbf{Action Modeling based Gradient Estimator} (AME) exploits jammer's action information $b_n$ to directly construct the entire cost vector $\U[:,b_n\mid\hatt]\in\mathbb{R}^{|\mathcal{A}_R|}$,
whose $a$-th component quantifies the cost that \emph{would be incurred} if the radar had played action $a$ against the observed jammer action $b_n$. Using this vector as $\hatl_n$ effectively turns the feedback from a single scalar observation into a full-information update over all radar actions, enabling OMD to downweight all unfavorable actions simultaneously.

By integrating the proposed \textbf{Action Modeling based Gradient Estimator} (AME) $\hatl_n$ into the OMD update~(10), we obtain the enhanced algorithm OMD-AME, summarized in Algorithm~\ref{alg:omd-ame}. Specifically, in each round, Algorithm~\ref{alg:omd-ame} alternates between (i) computing the gradient estimate $\hatl_n$ from the observed information $b_n$ via Definition~\ref{def:ame} and (ii) subsequently applying the standard OMD step driven by $\hatl_n$ to generate the strategy for the next round.


\setcounter{algocf}{0}

\begin{algorithm}
\caption{OMD with Action Modeling Estimator (OMD-AME)} \label{alg:omd-ame}
\SetAlgoLined
\KwIn{learning rate $\eta$, radar's initial strategy $\x_1$, cost matrix $\U(\hatt)$;}

\For{$n \gets 1$ \KwTo $N-1$}{
    Transmit signal based on action $a_n \sim \x_n$\;  
    Receive signal and extract jammer's action $b_n$\;   
    \textbf{Gradient Estimation}: $\hatl_n = \U{[:,b_n\mid\hatt]}$;\\
    \textbf{Online Mirror Descent}:
    \begin{equation*}
    \begin{cases}
    \text{Update: }&\boldsymbol{u}_{n+1}(a)=\x_{n}(a)e^{-\eta\hatl_{n}(a)},\forall a\in\ar;\\
    \text{Project: }&\x_{n+1}=\frac{\boldsymbol{u}_{n+1}}{\lVert\boldsymbol{u}_{n+1}\rVert_1}.
    \end{cases}  
    \end{equation*}
}
\end{algorithm}

Theorem~\ref{thm:omd-ame} gives the regret bound of OME-AME with detailed proof presented in Appendix~\ref{proof:thm1}. Compared to OMD-IWE (Exp3) in E.q.~\eqref{equ:iwe}, the proposed OMD-AME eliminates the square-root dependency on the action set $\ar$ due to the utilization of the inherent \textbf{matrix game structure within the cost function}. Additionally, it matches the lower bound $\Omega(\sqrt{N})$ of static regret mentioned in Definition~\ref{def:static}.
\begin{thm}[Regret bound for OMD-AME]\label{thm:omd-ame}
Denote $N$ as the total number of rounds, then OMD-AME algorithm guarantees the following regret bound:
\begin{equation*}\label{equ:ame-bound}
    \text{S-Regret}_N(\text{OMD-AME})\leq 2\sqrt{(\sigma^2+1)N\log{|\ar|}}
\end{equation*}
\end{thm}

\noindent\textbf{Proof sketch.} The regret analysis builds on the standard OMD framework, while adding adjustments for the bandit setting. The proof proceeds two steps: (i) the regret is first upper-bounded via the usual OMD decomposition, with the corresponding bound further extended by replacing the unknown true gradient with the proposed AME estimator; (ii) the unbiasedness and variance control of AME are then leveraged to linearize and bound the resulting estimation terms, leading to the final regret guarantee.
\begin{rmk}
    Compared to OMD-IWE, OMD-AME eliminates the $\sqrt{|\ar|}$ factor by effectively utilizing information about the jammer's actions and the cost structure. This approach enhances the efficiency of anti-jamming strategy learning and significantly improves sample efficiency.
\end{rmk}

A simulation example comparing OMD-IWE and OMD-AME is presented in Fig.~\ref{fig:alg-01}. In this figure, stationary refers to the scenario where the jammer's strategy $\y_n$ remains constant across all rounds $n$, while non-stationary indicates that $\y_n$ changes over time. Detailed settings are described in Section~\ref{sec:exp-setting}. In a stationary environment, OMD-AME reaches a higher SINR more quickly than OMD-IWE, highlighting its improved sample efficiency. However, in a non-stationary environment, neither OMD-IWE nor OMD-AME achieves high SINR. This is not surprising since the inherent \textbf{structure within the jammer's strategy} $\y_n$ has not been exploited. Next, more prior information on the decision-making behavior of the jammer will be utilized to enhance OMD-AME.
\begin{figure}
    \centering
    \includegraphics[width=0.95\linewidth]{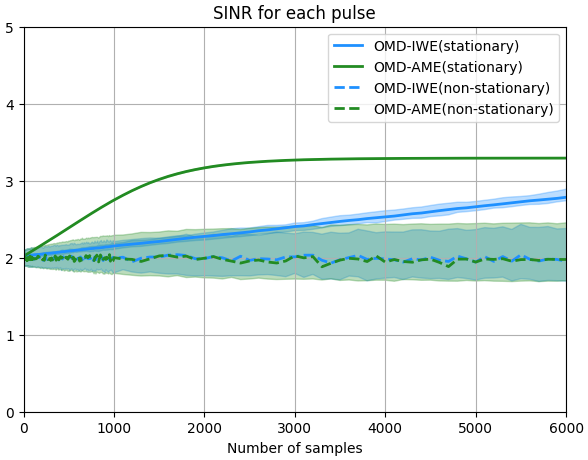}
    \caption{SINR comparison for OMD-IWE and OME-AME in stationary and non-stationary environments.}
    \label{fig:alg-01}
\end{figure}

\subsection{Opponent Modeling}\label{subsec:opponent}
Modern jamming system usually contains a crucial subsystem known as Digital Radio Frequency Memory (DRFM)~\cite{roome1990digital, berger2003digital}, which stores useful interaction histories that jammers can utilize to make informed decisions. In the following section, we focus on dealing with such a jammer by assuming that the jammer's strategy $\y_n$ follows a pre-defined \textbf{fixed decision rule} based on a length-$k$ history: 
$$h_n = \{a_{n-1}, b_{n-1}, \ldots, a_{n-k}, b_{n-k}\}\in\mathcal{H}=\mathcal{A}_\mathrm{R}^k\times\mathcal{A}_\mathrm{J}^k.$$
Under this setting, jamming strategy $\y_n$ is re-expressed as 
$$\y_n =\pi\left(h_n\right),$$
where $\pi:\mathcal{H} \rightarrow\Delta_{|\mathcal{A}_\mathrm{J}|}$ is defined as a \textbf{fixed} mapping from history space to the jammer's action space.

The history-dependent nature of the jammer’s behavior provides an additional exploitable structure beyond the realized action in a single round. Motivated by this, a new estimator termed the \textbf{Opponent Modeling based Gradient Estimator} (OME) is introduced, which infers $\y_n$ from a length-$k$ interaction history and maps it to the gradient through the modeled cost matrix. The detailed definition is shown below.
\begin{defi}[Opponent modeling gradient estimator]\label{def:ome}
Let $\hatt$ represent the estimated scenario parameters and $\hat{\pi}_n$ be an unbiased estimation of $\pi$ at round $n$. An unbiased gradient estimator, conditioned on the length-$k$ history $h_n$, is defined as:
\begin{equation}\label{equ:ome}
    \hatl_n = \U(\hatt)\hat{\pi}_n(h_n)
\end{equation}
\end{defi}
In this definition, the proposed gradient estimator $\hatl_n$ integrates estimations from both $\U(\theta)$ and $\pi$.
\paragraph{\textbf{Estimation of} $\U(\theta)$} As described in Remark~\ref{rmk:theta}, $\theta$ can be estimated as $\hatt$ with bounded variance.
\paragraph{\textbf{Estimation of} $\pi$} Maximum likelihood estimator (MLE) is employed to estimate $\pi$. At round $n$, MLE calculates the likelihood of each action $b \in \aj$, conditioned on the given $h$, across the entire interaction history $I_n=(a_1, b_1, \ldots, a_{n-1}, b_{n-1})$. Specifically, denote $h_i$ as the length-$k$ history in $I_{n}$ starting from round $i$, $h_i = \{ a_i,b_i,\ldots, a_{i+k-1}, b_{i+k-1} \}\subseteq I_n$, then $\forall b\in\aj$, if $h$ exists in history $I_n$,
\begin{equation}\label{equ:mle}
\begin{aligned}
    \hat{\pi}_n(b\mid h) = \frac{1}{\sum_{i=1}^{n-k-1}\mathbbm{1}(h_i=h)}\sum_{i=1}^{n-k-1} \mathbbm{1}(h_i=h,b_{i+k}=b)
\end{aligned}
\end{equation}
Otherwise, if $h$ does not exist in the history $I_n$, $\hat{\pi}_n(\cdot\mid h)$ follows uniform distribution. Notation $\mathbbm{1}(h_i=h)$ denotes the indicator function which equals to one when $h_i=h$ and is zero otherwise. 
The estimation of $\pi$ is unbiased, as $\mathbb{E}[\hat{\pi}_n(h_n)\mid I_n]=\pi(h_n)=\y_n$. Moreover, since $\hatt$ and $\hat{\pi}_n$ are independent, the proposed OME in E.q.~\eqref{equ:ome} is also unbiased.

Different from AME (Definition~\ref{def:ame}), which utilizes only the realized jammer action $b_n$ in the current round, OME yields a gradient estimator that predicts the gradient of the cost function ($\nabla_{\x} \phi(\x,\y)=\U\y$) under the history-dependent assumption.  Consequently, OME enables more accurate decision-making. 

By incorporating the proposed \textbf{Opponent Modeling Estimator} (OME) into the OMD framework described in E.q.(10), we obtain the OMD-OME algorithm (described in Algorithm~\ref{alg:omd-ome}. Specifically, in each round, Algorithm~\ref{alg:omd-ome} alternates between (i) observing the current length-$k$ history $h_n$ and determining its corresponding index $i_n=\mathbf{idx}(h_n)$, (ii) computing the opponent-modeling estimate $\hat{\pi}(H_{i_n})$ and forming the gradient estimate $\hat{l}_n$ via Definition 5, and (iii) applying the standard OMD update driven by $\hat{l}_n$ to update the associated strategy $\x^{(i_n)}$ for the next time the same history pattern is encountered.

As shown in Theorem~\ref{thm:omd-ome} below, the OMD-OME algorithm achieves sublinear universal regret, with a bound dependent on the size of the history space $\mH$ and the number of rounds $N$. Compared to the previously considered static regret, the universal regret defined in Eq.~\eqref{equ:u-regret} is a more rigorous performance metric, and achieving a sublinear bound on it represents a significant advancement in anti-jamming strategy design. This improvement offers an efficient mechanism to keep pace with the rule-based jammer. Further details are provided in the subsequent section.

\begin{algorithm}
\SetAlgoLined
\SetKwInOut{Input}{Input}
\SetKw{KwInit}{Initialize}
\SetKwFunction{FStrategyUpdate}{StrategyUpdate}
\caption{OMD with Opponent Modeling Estimator (OMD-OME)} \label{alg:omd-ome} 
\KwIn{ 
     learning rate $\eta$,  cost matrix $\U(\hatt)$, radar's action set $\ar$, length-$k$ history space $\mathcal{H}=\{H_i\}_{i=1}^{|\mathcal{H}|}$, initial history $I_0$, index mapping $\mathbf{idx}:\mathcal{H}\to[|\mathcal{H}|]$}
\KwInit{$: $}

    $\x^{(i)}\gets \frac{1}{|\ar|}\mathbf{1}, \forall i\in [|\mH|]$\tcp*{\footnotesize $\x^{(i)}$ is radar's strategy used with history $H_i$.}
    $C_i \gets 0, \forall i\in [|\mH|]$\tcp*{\footnotesize $C_i$ counts how many times $H_i$ has been observed so far.}
    $I\gets I_0$\tcp*{\footnotesize Entire history.}
    
\For{$n \gets 1$ \KwTo $N-1$}{
    Observe current length-$k$ history $h_n\in H$ from $I$\;
    $i_n\gets\mathbf{idx}(h_n)$\;
    Transmit signal based on action $a_n \sim \x^{(i_n)}$\;  
    Receive signal and extract jammer's action $b_n$\;
    Update entire history $I \gets I \cup \{b_n\}$\;
    
    $\x^{(i_n)}\gets$  \FStrategyUpdate{$h_n$, $\x^{(i_n)}$, $I$, $C_i$}\;
    $C_{i_n} \gets C_{i_n} + 1$\;

}
\hrulefill

\SetKwProg{Fn}{Function}{:}{}
\Fn{\FStrategyUpdate{$H$, $\x$, $C$, $I$}}{
\tcp{$C$ is the number of previous updates for history $H$}

        \textbf{Update}: $\hat{\pi}(H) \gets$ E.q.(16) with $I$\;
        
        \textbf{Gradient Estimation}:
            $\hatl= \mathbf{U}(\hatt) \hat{\pi}(H)$\;
        
         \For{$a\in\ar$}{
        $\boldsymbol{u}(a) \gets \x(a)\exp\!\big(-\eta\,\hatl(a)\big)$\;
    }
    $\x \gets \boldsymbol{u}/\|\boldsymbol{u}\|_1$\;
    \Return $\x$\;
}
\end{algorithm}

\subsection{Universal Regret Bound for OMD-OME}


The static regret defined in Definition~\ref{def:static} is not an appropriate theoretical performance metric for OMD-OME, as it compares $\x_n$ to a fixed optimal strategy $\x^*$, thereby implicitly assuming that a single fixed decision is effective across all rounds. To address this limitation, we adopt a more general metric called universal regret, which enables more flexible evaluations by utilizing a comparator sequence $\mathcal{Z}_N = (\z_1, \ldots, \z_N)$ that can change over time.
\begin{defi}[Universal regret] The universal regret of an online algorithm \(\mathcal{A}\) is defined as
\begin{equation}\label{equ:u-regret}
    \text{U-Regret}_N(\mathcal{A})=\mathbb{E}\left[\sum_{n=1}^Nf_n(\x_n) - \sum_{n=1}^N f_n(\z_n)\right],
\end{equation}
where $\z_n\in\mathcal{Z}_N,\forall n\in[N]$, and the expectation is taken over the randomness of the cost functions.
\end{defi}
The advantage of universal regret is its flexibility in selecting the comparator sequence $\mathcal{Z}_N$. For example, by setting $\z_1=\ldots=\z_N=\arg\min_{\x\in\dar} \sum_{n=1}^N f_n(\x)$, universal regret reduces to static regret, comparing against the best fixed strategy in hindsight. Alternatively, the optimal sequence $\mathcal{X}^*_N$ can be employed as the comparator, defined as
\begin{equation*}\label{equ:def-x^*}
    \mathcal{Z}_N = \mathcal{X}^*_N=(\x^*_1,\ldots,\x^*_N), \text{where } \x^*_n = \arg\min_{\x}f_n(\x),\forall n
\end{equation*}
Using $\mathcal{X}^*_N$ as the comparator aligns universal regret with the objective function defined in E.q.~\eqref{equ:obj}. Hence, achieving sublinear universal regret with respect to $\mathcal{X}^*_N$ indicates that the algorithm’s cumulative performance remains comparable to the per-round optimal solutions, thereby highlighting the strength of universal regret in dynamic environments.




Regarding the anti-jamming task, an important feature is that the jammer system also has a finite action set and its own strategy $\y_n$. By substituting $f_n(\x)$ with the two-player cost function $\phi(\x,\y)$ defined in E.q.~\eqref{equ:phi}, the universal regret w.r.t. $\mathcal{X}^*_N$ becomes
\begin{equation*}
    \text{U-Regret}_N(\mathcal{A})= \sum_{n=1}^N \phi(\x_n, \y_n) - \sum_{n=1}^N \phi(\x^*_n, \y_n)
\end{equation*}
This formulation highlights the influence of the jammer's strategy $\y_n$ on the universal regret. Furthermore, the estimation of $\y_n$ has been included in the proposed OME as defined in E.q.~\eqref{equ:ome}. Below, the universal regret bound of OMD-OME is given and a detailed proof is provided in Appendix~\ref{thm-proof:omd-ome}.


\begin{thm}[Regret bound for OMD-OME]\label{thm:omd-ome}
Let $N$ represent the total number of rounds, then the OMD-OME algorithm ensures a universal regret bound against the optimal comparator sequence $\mathcal{X}^*_N=(\x^*_1,\ldots,\x^*_N)$:
\begin{equation}\label{equ:ome-bound}
    \text{U-Regret}_N(\text{OMD-OME})\leq 2\sqrt{(\sigma^2+1) N|\mH|\log{|\mathcal{A}_R|}}
\end{equation}
where $|\mH|$ is the size of the history space.
\end{thm}

\noindent\textbf{Proof sketch.} The proof of the universal regret bound for OMD-OME is built on the OMD-AME analysis, but introduces two additional ideas: (i) partitioning the interaction into history-dependent subproblems, and (ii) fixing an optimal strategy for each subproblem. The first idea introduces history-based decomposition, while allowing us to bound the aggregation of per-history regret bounds. The second idea treats each history-dependent subproblem as a static-regret problem, enabling us to define a per-round optimal strategy $\x^*_n$ that in total supports a sublinear universal regret guarantee.

Incorporating prior knowledge about the jammer, i.e., finite action space $\mathcal{B}$ and jamming strategy generation rule $\pi$, into the anti-jamming strategy updates is essential for achieving a sub-linear bound on the more challenging universal regret. This approach allows for finding the optimal strategy in each round $n$, thereby significantly enhancing sample efficiency and improving performance in non-stationary environments.

%% file: sections/expe.tex
In this section, we evaluate the performance of the proposed OMD-AME and OMD-OME algorithms by simulations.
\subsection{Simulation Settings}\label{sec:exp-setting}
\noindent\textbf{Simulation Environment:} To ensure a high degree of realism, a signal-level simulation environment was developed using MATLAB and its associated toolboxes. This environment simulates the transmitted radar signal, spatial propagation effects, jamming interference, and target reflections, with the primary objective of closely mirroring practical operational scenarios. The proposed algorithms are evaluated within this environment to closely approximate real-world conditions. Detailed system parameters are provided in Table~1, which contains key radar, jammer, and propagation parameters (e.g., antenna gains, transmit powers, carrier frequency, available sub-frequencies, and distances). In each round, the radar and jammer transmit waveforms that propagate according to a standard free-space path-loss model. Furthermore, all simulation source code has been made publicly available on GitHub\footnote{\href{https://github.com/ptponway/radar_platform}{Source code for the anti-jamming platform.}}.

A set of simulated single-round examples within the environment is illustrated in Fig.~\ref{fig:4spectrum}, showcasing spectrograms of the radar alongside three different jammer types. The SINR values depicted in these figures are computed using the post-processing method described in Appendix~\ref{appdix:postprocessing}. Furthermore, as defined in E.q.~\eqref{equ:sinr}, the SINR values obtained in each round are utilized by the algorithms to update their respective anti-jamming strategies. Regarding jammer modeling, three different types of jamming rules are incorporated into the simulation environment, each of which is detailed in the following subsections.
\begin{figure}
    \centering
    \includegraphics[width=0.99\linewidth]{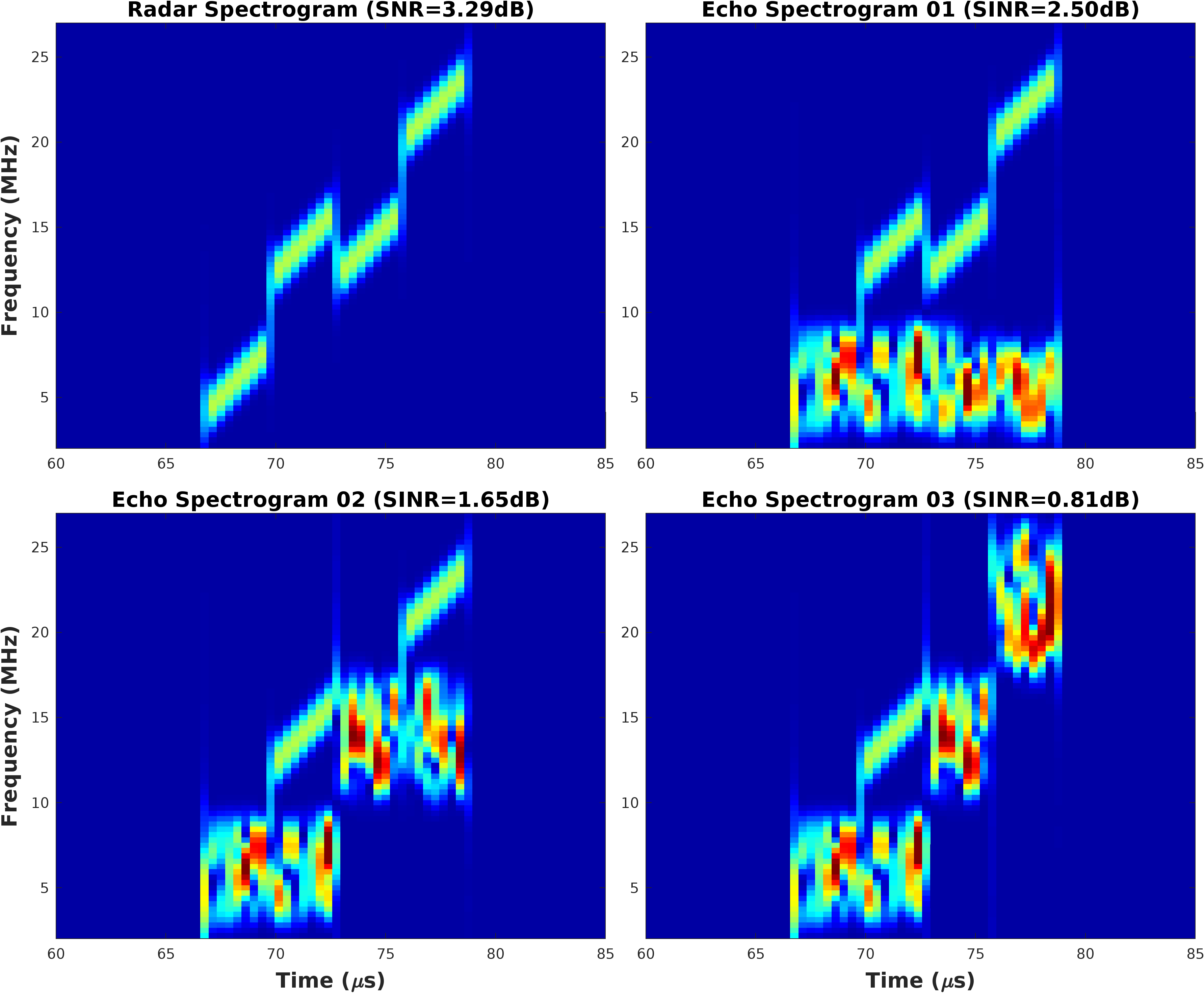}
    \caption{Simulated single-round spectrograms of radar and three different jamming signals.}
    \label{fig:4spectrum}
\end{figure}

\noindent\textbf{Performance Metric:} To account for the stochastic nature of the algorithms, $1000$ independent trials are conducted. The static regret in \eqref{equ:sta-regret}, universal regret in \eqref{equ:u-regret} and SINR in \eqref{equ:sinr} averaged over trials are utilized as performance metrics. Specifically, from Fig.~\ref{fig:exp-sta01} to Fig.~\ref{fig:sinr}, each solid line represents the mean value across all $1000$ trials, and the corresponding shaded regions indicate the $95\%$ confidence intervals.

\noindent\textbf{Baselines:} 
We benchmark the proposed algorithms against two established baseline methods. The first is a classical OCO algorithm, Exp3 (renamed as OMD-IWE in our study), which does not take into account the structure information of the opponent. The second is a widely-used reinforcement learning algorithm, the Deep Q-Network (DQN). This method has been extensively employed in various existing studies on anti-jamming strategy learning~\cite{liRadarActiveAntagonism2021, liDeepQNetworkBased2019, zhengAirborneRadarAntiJamming2022}.

\begin{table}[t]
\centering
\begin{tabular}{|c|c|c|}
\hline
Parameter          &Notation   & Value      \\ \hline
Radar antenna gain        & $G_R$        & 30dB   \\
Radar transmitted power      & $P_R$             & 10KW       \\
Jammer transmitted power     & $P_J$              & 1KW        \\
\# of sub-pulses    & $M$     & 4         \\
Available frequencies & $\{f_1,f_2,f_3\}$ & $\{6,14,22\}$MHz  \\
Sub-pulse width      & $\tau$  & 3$\mu$s   \\
Pulse repetition frequency    & PRF   & 1000Hz      \\
Carrier frequency    & $f_c$  & 10GHz     \\
Distance             & $R$   & 100KM      \\ \hline
\end{tabular}
\caption{Parameter settings for FA radar and jammer.}
\label{tab:para}
\end{table}

\subsection{Comparison in a Stationary Environment}\label{exp:stationary}
In this subsection, we analyze a simple jammer that targets specific predefined frequencies. Under this scenario, the jamming strategy $\y_n$ remains constant across all rounds, characterizing what we refer to as a stationary environment. Specifically, $\y_n$ is defined as follows:
\begin{equation*}
\y_n(b)=
\begin{cases}
0.2, & b= [f_1, f_1, f_2, f_2] \\
0.4, & b = [f_1, f_1, f_1, f_1] \\
0.4, & b = [f_2, f_2, f_2, f_2] \\
0, &\text{Otherwise}
\end{cases}
\end{equation*}

Given the stationary nature of the environment, static regret and universal regret coincide. Fig.~\ref{fig:exp-sta01} illustrates the performance of various methods in terms of average static (universal) regret, with both axes plotted on a logarithmic scale. A faster decline in average regret indicates more effective convergence to the optimal strategy. By the $10^4$-th iteration, the proposed OMD-OME method achieves the lowest regret, reducing it below $10^{-3}$. The OMD-AME method attains a higher regret level of approximately $10^{-1}$. In contrast, baselines methods such as OMD-IWE (Exp3) and DQN incur significantly higher regret, approaching values around 
0.7. Furthermore, the DQN method demonstrates a notably slower reduction in regret, even after $10^6$ iterations, primarily due to its limited sample efficiency. The superior performance of OMD-AME and OMD-AME can be attributed to their efficient exploration of the action space and rapid adaptation to the stationary jammer strategy.
\begin{figure}
    \centering
    \includegraphics[width=0.8\linewidth]{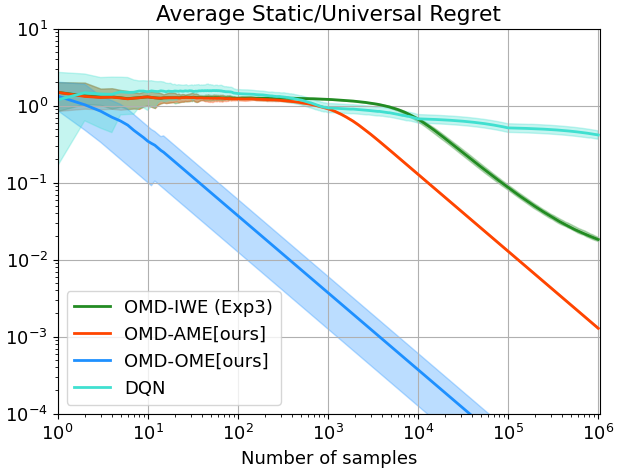}
    \caption{Regret comparison in a stationary environment.}
    \label{fig:exp-sta01}
\end{figure}

\subsection{Comparison in Non-stationary Environments}\label{exp:non-stationary}
In this section, we examine the performance of the proposed algorithms under more complex scenarios, where the jammer's strategy $\y_n$ changes over rounds and we refer to this scenario as non-stationary jamming environment. Specifically, $\y_n$ is generated via a pre-defined decision rule $\pi(\cdot)$ based on a fixed-length history (c.f. Section~\ref{subsec:opponent}). Two types of rules are considered: $\pi(\cdot)$ is binary which implies a deterministic rule and $\pi(\cdot)$ is a non-binary probability simplex that corresponds to a stochastic rule.

\noindent\textbf{Deterministic $\pi(\cdot)$}: We evaluate a smeared spectrum jammer (SMSP)\cite{sun2009suppression}, which selects frequencies based on specific positions in the radar's recent length-$2$ history. Specifically,
\begin{equation*}
\pi_n(b\mid h)=
\begin{cases}
1 & b=b_{\text{chosen}} \\
0 & b\neq b_{\text{chosen}}
\end{cases}
\end{equation*}
For instance, $b_{\text{chosen}}$ integrates frequencies from the first and third positions of the radar's two most recent pulses. If the recent two pulses are $[f_1,f_1,f_2,f_2]$ and $[f_3,f_3,f_3,f_1]$, then $b_{\text{chosen}}=[f_1,f_2,f_3,f_3]$.

Fig.~\ref{fig:sta-det2} demonstrates that all algorithms achieve sub-linear average static regret, consistent with findings in stationary environments. Notably, the proposed OMD-OME and OMD-AME algorithms exhibit faster regret reduction compared to the baseline methods. OMD-OME, in particular, shows a dramatic decline in average static regret around $600$ rounds, indicating substantial convergence improvements over the baselines. Thus, OMD-OME closely matches or even surpasses the optimal static strategy in cumulative cost under the non-stationary environment, highlighting its remarkable adaptability. Additionally, the baseline DQN can also outperform the optimal static strategy given sufficient samples (e.g., $10^5$ rounds). For the more challenging metric of average universal regret (Fig.~\ref{fig:uni-det2}), only OMD-OME and DQN achieve sub-linear regret trajectories. OMD-OME reduces universal regret rapidly, requiring only thousands of rounds to achieve an average universal regret of $10^{-1}$, whereas DQN decreases much more slowly, reaching an average universal regret of approximately $0.8$ after even $10^6$ rounds.
\begin{figure}[htbp]
  \centering
  \begin{subfigure}[b]{0.49\textwidth}
    \centering
    \includegraphics[width=0.8\textwidth]{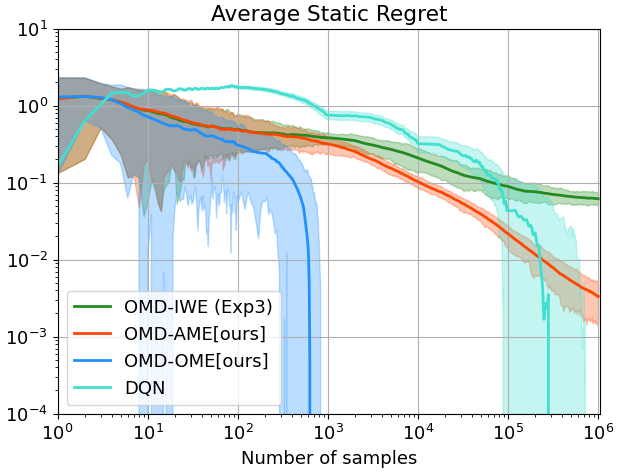}
    \caption{Static regret}
    \label{fig:sta-det2}
  \end{subfigure}
  \begin{subfigure}[b]{0.49\textwidth}
    \centering
    \includegraphics[width=0.8\textwidth]{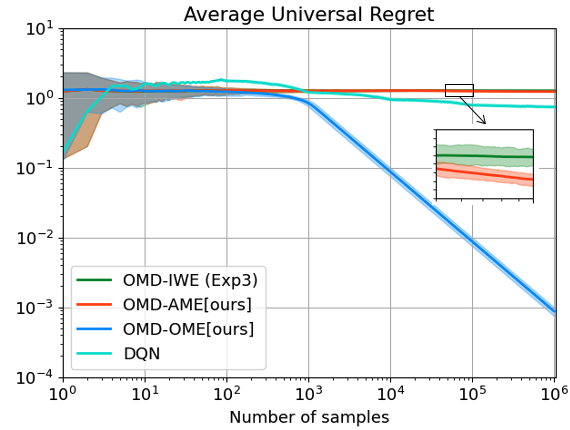}
    \caption{Universal regret}
    \label{fig:uni-det2}
  \end{subfigure}
\caption{Regret comparison in a non-stationary environment.}
\label{fig:non-stat01}
\end{figure}

\noindent\textbf{Stochastic $\pi(\cdot)$}: We consider a smart noise jammer (SNJ)\cite{zhang2022radar} that jams the two most commonly occurring frequencies in the recent length-$2$ history, denoted as $f_{m1}$ and $f_{m2}$. The jammer's rule $\pi(\cdot)$ is given by:
\begin{equation*}
\pi_n(b\mid h)=
\begin{cases}
0.7 & b=[f_{m1},f_{m1},f_{m1},f_{m1}] \\
0.3 & b=[f_{m2},f_{m2},f_{m2},f_{m2}] \\
0 & \text{Otherwise}
\end{cases}
\end{equation*}

Fig.~\ref{fig:freq2} demonstrates that the proposed OMD-OME and OMD-AME algorithms maintain strong static regret performance. For average universal regret, only OMD-OME achieves sub-linear performance within $10000$ samples. Although the RL baseline DQN also exhibits a decreasing trend in universal regret, it requires significantly more samples than OMD-OME to reach sub-linear universal regret. Further details are discussed in the subsequent section on sample efficiency.
\begin{figure}[htbp]
  \centering
  \begin{subfigure}[b]{0.49\textwidth}
    \centering
    \includegraphics[width=0.8\textwidth]{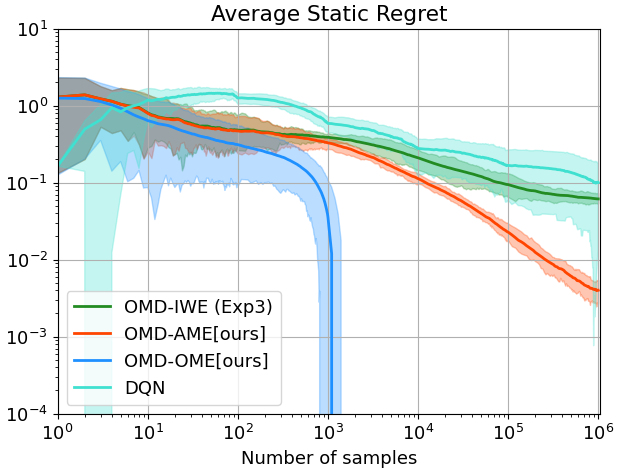}
    \caption{Static regret}
    \label{fig:sta-freq2}
  \end{subfigure}
  \begin{subfigure}[b]{0.49\textwidth}
    \centering
    \includegraphics[width=0.8\textwidth]{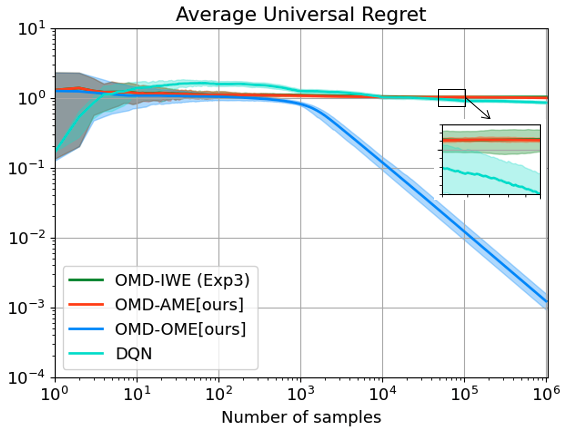}
    \caption{Universal regret}
    \label{fig:uni-freq2}
  \end{subfigure}
\caption{Comparison in non-stationary environment}
\label{fig:freq2}
\end{figure}

\subsection{Comparison on calculated SINR}
In this section, the SINR value (c.f. E.q.~\ref{equ:sinr}) after pulse compression is computed for each round, which provides a direct comparison of the anti-jamming performance across different algorithms. To ensure a comprehensive comparison, both the stationary and non-stationary environments are considered, set identically to those described in Sections~\ref{exp:stationary} and \ref{exp:non-stationary}.

Fig.~\ref{fig:sinr-sta0} indicates that all methods are capable of achieving optimal SINR in the stationary environment. Specifically, the proposed OMD-OME and OMD-AME algorithms demonstrate a rapid increase in SINR, reaching an average maximum value of $3.35$ dB within only around $10$ and $3000$ samples, respectively. This rapid performance underscores their efficiency in achieving high SINR with minimal sample usage. In contrast, the baseline OMD-IWE fails to achieve the maximum value within $10^4$ samples, and the baseline DQN is not stable. For the non-stationary environment compared in Fig.~\ref{fig:sinr-det2}, OMD-OME is the only algorithm that consistently reaches the average maximum SINR within near $1200$ samples. This highlights its superior performance and the benefits obtained by leveraging the characteristics of the jammer's information. In contrast, the baseline algorithms cannot achieve a stable optimal SINR value within $10^4$ rounds.
\begin{figure}[htbp]
  \centering
  \begin{subfigure}[b]{0.49\textwidth}
    \centering
    \includegraphics[width=0.8\textwidth]{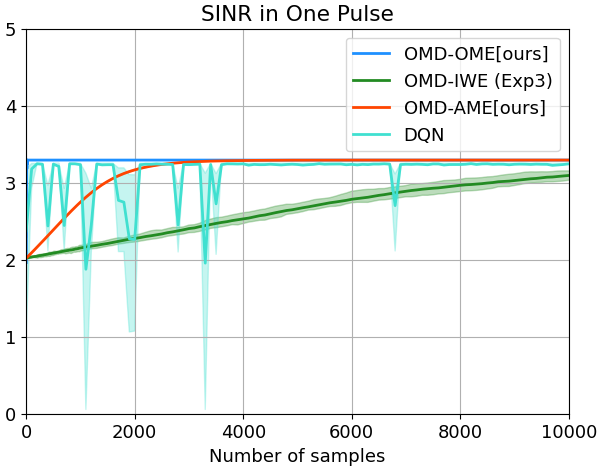}
    \caption{Stationary environment}
    \label{fig:sinr-sta0}
  \end{subfigure}
  \begin{subfigure}[b]{0.49\textwidth}
    \centering
    \includegraphics[width=0.8\textwidth]{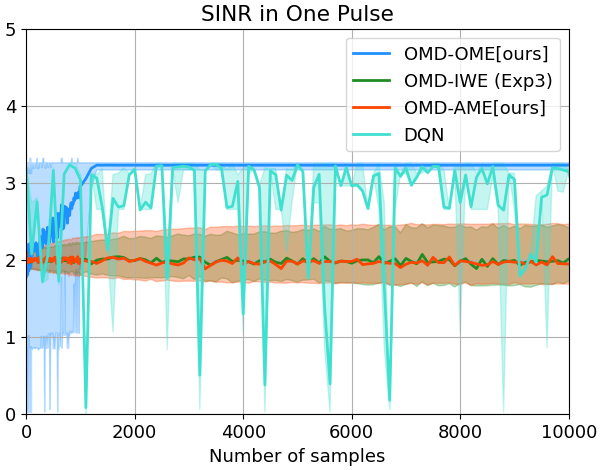}
    \caption{Non-stationary environment}
    \label{fig:sinr-det2}
  \end{subfigure}
\caption{SINR Comparison}
\label{fig:sinr}
\end{figure}

\begin{table*}
\centering
\begin{tabular}{|c|c|c|c|c|}
\hline
\diagbox{\textbf{Jamming environment}}{\textbf{Methods}}                                       & OME-IWE (Exp3) & DQN           & OMD-AME    & OMD-OME       \\ \hline
Stationary - Fixed $\y_n $                          & $9\times10^3$ & $4\times10^4$ & $1.5\times 10^3$      & $8$          \\ \hline
Non-stationary - Deterministic $\pi(\cdot)$ & \slashbox{}{}    & $2\times10^5$ & \slashbox{}{} & $1.2\times10^3$         \\ \hline
Non-stationary - Stochastic $\pi(\cdot)$    & \slashbox{}{}    & $1.2\times10^6$ & \slashbox{}{} & $2.2\times10^3$ \\ \hline
\end{tabular}
\caption{Approximated number of samples required for consistently achieving $3$ dB SINR.}
\label{tab:sample}
\end{table*}

\subsection{Sample Efficiency}
To provide a clear comparison of sample efficiency, Table~\ref{tab:sample} displays the approximated number of samples required to consistently achieve a SINR value of $3$ dB, averaged over $1000$ independent trails. Notably, the proposed OMD-OME and OMD-AME methods demonstrate significant efficiency, requiring only $8$ and $1500$ samples in the stationary environment, respectively. In contrast, the baseline algorithms OMD-IWE and DQN require substantially more samples, approximately $9,000$ and $4 \times 10^4$, respectively. In the non-stationary environment, only the OMD-OME and the RL baseline DQN eventually manage to reach the high SINR value of $3$ dB. Among these, the proposed OMD-OME method stands out, requiring only $2200$ samples in the non-stationary environment with a stochastic $\pi$, while the baseline DQN requires millions of samples. This highlights the substantial improvement in sample efficiency achieved by the proposed algorithms.

%% file: sections/appendix.tex
\subsection{Post-processing Details}\label{appdix:postprocessing}
For the $n$-th pulse with sub-frequencies $f^\mr_m, \forall m\in[M]$, the received signal undergoes bandpass filtering at each frequency:
\begin{equation*}
    Y^{\text{bp}}_n(e^{j\omega}) = R_n(e^{j\omega})H_n(e^{j\omega}), \omega=2\pi f^\mathrm{R}_m/f_s,m\in[M]
\end{equation*}
where $R_n(e^{j\omega})$ and $H_n(e^{j\omega})$ represent the Fourier transform of $r_n(t)$ and the filter's frequency response, respectively. Subsequently, matched filtering is applied to detect the jamming signal  for each corresponding $Y^{\text{bp}}n(e^{j\omega})$ components:
\begin{equation*}
\small
    Y^{\text{mf}}_n(e^{j\omega}) = Y^{\text{bp}}_n(e^{j\omega})S^{\text{bp},*}_n(e^{j\omega}), \omega=2\pi f^\mathrm{R}_m/f_s,m\in[M]
\end{equation*}
where $S^{\text{bp}}_n(e^{j\omega})$ is the impulse response after bandpass filtering. The environment parameter vector $\theta=[P_\mr, P_\mj, P_{n_0}]$ is then estimated from $Y^{\text{mf}}_n(e^{j\omega})$ through post-processing.

\subsection{Proof for Lemma~\ref{lem:grad-est}}\label{proof:lem-grad}
\begin{proof}
Define the gradient estimator as $\hatl_n=\rl_n+\Delta\rl_n$. The static regret defined in E.q.\eqref{equ:sta-regret} over $N$ rounds for algorithm $\mathcal{A}$ can be expressed as
\begin{equation*}
\begin{aligned}
    \text{S-Regret}_N(\mathcal{A}) 
    =& \mathbb{E}\left[ \sum_{n=1}^N \langle\rl_n,\x_n\rangle - \sum_{n=1}^N \langle\rl_n,\x^*\rangle \right] \\
    =& \mathbb{E}\left[ \sum_{n=1}^N \langle\hatl_n,\x_n\rangle - \sum_{n=1}^N \langle\hatl_n,\x^*\rangle \right] \\
    &-\underbrace{ \mathbb{E}\left[ \sum_{n=1}^N \langle\Delta\rl_n,\x_n-\x^*\rangle \right]}_{\textbf{error term}\ \epsilon_N}
\end{aligned}
\end{equation*}
If $\hatl_n$ is an unbiased estimator, which means $\mathbb{E}[\Delta\rl_n]=0,\forall n$, the error term vanishes as $\epsilon_N=0$. Therefore, the static regret can be entirely formulated by estimators $\hatl_n$.

\end{proof}

\subsection{Proof of Theorem~\ref{thm:omd-ame}}\label{proof:thm1}
\begin{proof}
\textbf{(Applying gradient estimator)} The static regret, as defined in E.q.\eqref{equ:sta-regret}, can be expressed using the proposed gradient estimator $\hatl_n=\U(\hatt)\boldsymbol{b}_n$ through the following steps:
\begin{equation*}
\begin{aligned}
    \text{S-Regret}_N(\mathcal{A})
    &= \mathbb{E}\left[ \sum_{n=1}^N \langle\rl_n,\x_n\rangle - \sum_{n=1}^N \langle\rl_n,\x^*\rangle \right] \\
    &=\mathbb{E}\left[ \sum_{n=1}^N \langle\hatl_n,\x_n\rangle - \sum_{n=1}^N \langle\hatl_n,\x^*\rangle \right] \\
    &= \mathbb{E}\left[ \sum_{n=1}^N \langle\hatl_n,\x_n\rangle - \sum_{n=1}^N \hatl_n(a^*) \right]
\end{aligned}
\end{equation*}
Here, $\rl_n=\U(\theta)\y_n$ is the gradient with respect to $\x_n$. Lemma~\ref{lem:grad-est} ensures the validity of the second equality in the derivation. $\x^*$, as defined in E.q.~\eqref{equ:x^*}, is a one-hot vector, and $a^*\in\ar$ represents the corresponding optimal action.

\begin{lem}\label{lem:hedge}
Let $\x_n$ be generated by OMD-AME with a constant stepsize $\eta>0$ and $\hatl_n$ be a gradient estimator, then 
\begin{equation*}\label{equ:inequality-hedge}
    \sum_{n=1}^N \langle\hatl_n,\x_n\rangle \leq \sum_{n=1}^N \hatl_n(a^*) + \eta\sum_{n=1}^N\langle\hatl_n^{\,2},\x_n\rangle+\frac{\log{|\ar|}}{\eta},
\end{equation*}
where $\hatl^2_n$ denotes the element-wise squaring of the vector $\hatl_n$, i.e., $\hatl^2_n(a)=\hatl_n(a)^2,\forall a\in\ar$.
\end{lem}
The proof is provided in Appendix~\ref{proof:lem-hedge}. Utilizing Lemma~\ref{lem:hedge}, the static regret can be bounded as follows:
\begin{equation*}\label{equ:static_proof02}
\text{S-Regret}_N(\mathcal{A})\leq \mathbb{E}\left[\eta\sum_{n=1}^N\langle\hatl_n^{\,2},\x_n\rangle+\frac{\log{|\ar|}}{\eta}\right].
\end{equation*}
\noindent\textbf{(Bounding $\mathbb{E}\left[\langle\hatl_n^{\,2},\x_n\rangle\right]$)}
Since $\hatl_n$ only depends on $b_n$, $\x_n$ and $\hatl^2_n$ are independent, and $\mathbb{E}\left[\langle\hatl_n^{\,2},\x_n\rangle\right]$ can be decomposed as
\begin{equation*}
\mathbb{E}\left[\langle\hatl_n^{\,2},\x_n\rangle\right]
=\sum_{a=1}^{|\ar|} \mathbb{E}[\x_n(a)]\mathbb{E}[\hatl^2_n(a)]
\end{equation*}
According to Remark~\ref{rmk:theta}, for each entry in the cost matrix $\U$,
\begin{equation*}
   \mathbb{E}\left[\U[a,b]\right]\leq1, \quad \text{Var}\left[\U[a,b]\right]\leq\sigma^2 
\end{equation*}
Since $\hatl_n(a)=\U[a,b_n], \forall a\in\ar$, we have:
\begin{equation}\label{equ:Usquare}
    \mathbb{E}\left[ \hatl_n(a)^2 \right]=\mathbb{E}\left[\hatl_n(a)\right]^2+\text{Var}\left[\hatl_n(a)\right] 
    \leq\sigma^2+1
\end{equation}
Considering that $\x_n$ is a probability distribution,
\begin{equation*}\label{equ:sta-proof03}
\mathbb{E}\left[\langle\hatl_n^{\,2},\x_n\rangle\right]\leq(\sigma^2+1)\mathbb{E}\left[\sum_{a=1}^{|\ar|}\x_n(a)\right]=\sigma^2+1
\end{equation*}
\noindent\textbf{(Final regret bound)} Consequently, the bound on the static regret is derived as follows:
\begin{equation*}
\begin{aligned}
    \text{S-Regret}_N(\mathcal{A})&\leq\mathbb{E}\left[\eta\sum_{n=1}^N\langle\hatl_n^{\,2},\x_n\rangle+\frac{\log{|\ar|}}{\eta}\right] \\
    &\leq \eta(\sigma^2+1)N +\frac{\log{|\ar|}}{\eta} 
\end{aligned}    
\end{equation*}
Set $\eta=\sqrt{\log{|\mathcal{A}_R|}/(\sigma^2+1)N}$, from which it follows that:
\begin{equation*}
\text{S-Regret}_N(\mathcal{A})\leq 2\sqrt{(\sigma^2+1)N\log{|\mathcal{A}_R|}}
\end{equation*}
\end{proof}

\subsection{Proof for Lemma~\ref{lem:hedge}}\label{proof:lem-hedge}
Recall the update rule for $\x_{n+1}$ described in Algorithm~\ref{alg:omd-ame}, then for any $a\in\ar$,
\begin{equation}\label{equ:xn}
\begin{aligned}
    \x_{n+1}(a) &= \frac{\x_n(a) e^{-\eta\hatl_n(a)}}{\sum_{i}\x_n(i) e^{-\eta\hatl_n(i)}}=\frac{1}{c_n}\x_n(a) e^{-\eta\hatl_n(a)},
\end{aligned}
\end{equation}
where $c_n=\sum_{i}\x_n(i) e^{-\eta\hatl_n(i)}$ represents the normalization constant at round $n$, and $\hatl_n$ is the gradient estimator. Through induction, we have
\begin{equation*}\label{equ:Ln}
\begin{aligned}
    \x_{n+1}(a) &= \left( \frac{1}{c_n}\frac{1}{c_{n-1}}\ldots\frac{1}{c_1}\x_1(a)\right) e^{-\sum_{j=1}^n\eta\hatl_j(a)} \\
    &\overset{\Delta}{=} p_{n} \bv_n(a) ,
\end{aligned}
\end{equation*}
where $\bv_n(a)=e^{-\sum_{j=1}^n\eta\hatl_j(a)}$ represents exponentially cumulative weighted reward for $a$ and $$p_{n}=\left( \frac{1}{c_n}\frac{1}{c_{n-1}}\ldots\frac{1}{c_1}\x_1(a)\right)=\lVert \bv_n \rVert_1$$ 
is a constant factor. Substituting $\x_n(a)=p_{n-1}\bv_{n-1}(a)$ into \eqref{equ:xn}, $\x_{n+1}(a)$ can be re-expressed as
\begin{equation}\label{equ:xnLn}
    \x_{n+1}(a) =\frac{\bv_n(a)}{\lVert \bv_n \rVert_1},\  \forall a\in\ar.
\end{equation}
For round $n$, by definition of $\bv_n(a)$ and E.q.~\eqref{equ:xnLn}, 
\begin{equation*}
\begin{aligned}
    \lVert \bv_{n}\rVert_1 &=\sum_{a\in\ar} e^{-\sum_{j=1}^n\eta\hatl_j(a)}
    = \sum_{a\in\ar}\bv_{n-1}(a)e^{-\eta\hatl_n(a)} \\
&= \lVert \bv_{n-1}\rVert_1\sum_{a\in\ar}\x_n(a)e^{-\eta\hatl_n(a)}
\end{aligned}
\end{equation*}
Then, we further bounded $\lVert \bv_{n}\rVert_1$ as
\begin{equation}\label{eq:vn}
\begin{aligned}
\lVert \bv_{n}\rVert_1&\leq \lVert \bv_{n-1}\rVert_1\sum_{a\in\ar} \x_n(a)\left( 1-\eta\hatl_n(a)+\eta^2\hatl_n(a)^2 \right) \\
&=\lVert \bv_{n-1}\rVert_1\left(1-\eta\langle\hatl_n,\x_n\rangle+\eta^2\langle\hatl_n^{\,2},\x_n\rangle\right) \\
&\leq \lVert \bv_{n-1}\rVert_1e^{-\eta\langle\hatl_n,\x_n\rangle+\eta^2\langle\hatl_n^{\,2},\x_n\rangle}
\end{aligned}
\end{equation}
where the first inequality comes from the Taylor series expansion that $e^{-x}\leq 1-x+x^2$, and the second inequality is from $1+x\leq e^x$. Through induction for E.q.~\eqref{eq:vn}, the relationship between $\bv_N$ and $\x_n$ can be expressed as
\begin{equation}\label{equ:proof-inequ02}
    \lVert \bv_{N}\rVert_1 \leq \lVert \bv_{1}\rVert_1e^{-\sum_{n=1}^N\eta\langle\hatl_n,\x_n\rangle+\sum_{n=1}^N\eta^2\langle\hatl_n^{\,2},\x_n\rangle}.
\end{equation}
Note that for the optimal action $a^*$,
\begin{equation}\label{equ:proof-inequ01}
    \bv_{N}(a^*)=e^{-\sum_{n=1}^N\eta\hatl_n(a^*)}\leq \lVert \bv_{N}\rVert_1.
\end{equation}
Then combine inequalities \eqref{equ:proof-inequ02} and \eqref{equ:proof-inequ01} together, we have
\begin{equation*}
    e^{-\sum_{n=1}^N\eta\hatl_n(a^*)}\leq |\ar|e^{-\sum_{n=1}^N\eta\langle\hatl_n,\x_n\rangle+\sum_{n=1}^N\eta^2\langle\hatl_n^{\,2},\x_n\rangle}
\end{equation*}
where $\lVert \bv_1\rVert_1\leq|\ar|$ is used. By taking the logarithm on both sides and re-organizing the order, we have
\begin{equation*}\label{equ:inequality-hedge}
    \sum_{n=1}^N \langle\hatl_n,\x_n\rangle - \sum_{n=1}^N \hatl_n(a^*)\leq \eta\sum_{n=1}^N\langle\hatl_n^{\,2},\x_n\rangle+\frac{\log{|\ar|}}{\eta},
\end{equation*}
where $\hatl_n^{\,2}$ represents element-wise squaring of the vector $\hatl_n$.


\subsection{Proof for Theorem.\ref{thm:omd-ome}}\label{thm-proof:omd-ome}
\begin{proof}
\textbf{(Grouping regrets)} Due to the fixed history-length property, the history space \(\mH\) can be re-expressed as a finite set of elements, i.e., \(\mH = \{H_1, \ldots, H_{|\mH|}\}\). This allows the universal regret, as defined in \eqref{equ:u-regret}, to be categorized into distinct groups corresponding to identical histories. Notably, for each $H\in\mathcal{H}$, the jammer's strategy $\y=\pi(H)$ is fixed, thereby creating a locally stationary environment. Thus, the entire universal regret can be decomposed as the summation of various static regrets, each conditioned on a distinct $H$. The detailed decomposition is provided in Eq.~\eqref{equ:uregret-decomp}. Furthermore, denote $\x^*_H$ as the optimal strategy associated with $H$, then reorganize E.q.~\eqref{equ:uregret-decomp}, we obtain
\begin{figure*}
\begin{equation}\label{equ:uregret-decomp}
\begin{aligned}
    \text{U-Regret}_N(\mathcal{A}) &= \sum_{n=1}^N \phi(\x_n, \y_n) - \sum_{n=1}^N \phi(\x^*_n, \y_n) \qquad\qquad\qquad\text{( Definition of universal regret )}\\ 
    &= \sum_{n=1}^N \phi\left(\x_n, \pi(h_n)\right) - \sum_{n=1}^N \phi\left(\x^*_n, \pi(h_n)\right) \quad\qquad\text{( History-dependent mapping }\pi\text{ )}\\
    &= \underbrace{\sum_{\substack{n\in[N]:\\h_{n}=H_1}} \phi\left(\x_{n}, \pi(H_1)\right) - \phi\left(\x^*_{H_1}, \pi(H_1)\right)}_{\text{Group with the same history }H_1} +\ldots+ \underbrace{\sum_{{\scriptsize\substack{n\in[N]:\\h_{n}=H_{|\mH|}}}} \phi\left(\x_{n}, \pi(H_{|\mH|})\right) -  \phi\left(\x^*_{H_{|\mH|}}, \pi(H_{|\mH|})\right)}_{\text{Group with the same history }H_{|\mH|}}
\end{aligned}
\end{equation}
\end{figure*}
\begin{equation*}\label{equ:proof-ome01}
\begin{aligned}
\text{\normalsize U-Regret}&_N\text{\normalsize$(\mathcal{A})$}=\sum_{n=1}^N \phi(\x_n, \y_n) - \sum_{n=1}^N \phi(\x^*_n, \y_n) \\
&=\sum_{i=1}^{|\mH|}\sum_{j\in \mathcal{N}_i} \phi\left(\x_{j}, \pi(H_i)\right) - \phi\left(\x^*_{H_i}, \pi(H_i)\right)
\end{aligned}
\end{equation*}
where $\mathcal{N}_i=\{ n\mid h_{n}=H_i,n\in[N]\}$ denotes the set of round indices that share the same history $H_i$, the first summation encompasses over all histories in $\mathcal{H}$, and the second summation aggregates the static regret for each individual history.

\noindent\textbf{(Applying gradient estimator)} Since the proposed opponent-modeling gradient estimator (OME) is also conditioned on the history $H\in\mathcal{H}$, it can be applied to each sub-regret term after the decomposition. Following a similar analysis to that of OMD-AME, we can directly obtain
\begin{equation}\label{equ:proof-uregret}
\begin{aligned}
\text{\normalsize U-Regret}_N\text{\normalsize$(\mathcal{A})$} &=\mathbb{E}\left[\sum_{i=1}^{|\mH|}\sum_{j\in\mathcal{N}_i}\langle\hatl_{j},\x_{j}\rangle-\langle\hatl_{j},\x^*_{H_i}\rangle \right] \\
&\leq\mathbb{E}\left[\sum_{i=1}^{|\mH|}\sum_{j\in\mathcal{N}_i}\eta\langle\hatl_{j}^{\,2},\x_{j}\rangle+\frac{\log{|\ar|}}{\eta}\right]
\end{aligned}
\end{equation}
where the expectation accounts for the randomness in the algorithm and history. The equality follows from Lemma~\ref{lem:grad-est}, and the inequality arises from Lemma~\ref{lem:hedge}.

\noindent\textbf{(Bounding Regret)} For each $H_i\in\mathcal{H}$, we need to further bound the expected inner product term $\mathbb{E}\left[\langle\hatl_{j}^{\,2},\x_{j}\rangle\right]$. Using similar techniques as in Appendix~\ref{proof:thm1}, for all $j\in\mathcal{N}_i$,
\begin{equation*}
\begin{aligned}
\mathbb{E}\left[\langle\hatl_{j}^{\,2},\x_{j}\rangle\right] = \mathbb{E}\left[ \sum_{a=1}^{|\ar|} \x_{j}(a)\hatl_{j}(a)^2 \right] =\mathbb{E}\left[ \sum_{a=1}^{|\ar|}\x_{j}(a) \mathbb{E}\left[ \left(\U[a,:]\hat{\y}_{j}\right)^2 \right] \right],
\end{aligned}   
\end{equation*}
where  $\hatl_{j}(a)=\U[a,:]\hat{\y}_{j}$ is used and the second equality holds since $\x_{j}$ and $\hatl_{j}$ are independent. Furthermore,
\begin{equation*}
\begin{aligned}
    \mathbb{E}\left[\left(\U[a,:\mid\hat{\theta}]\hat{\y}_{j}\right)^2\right] &= \mathbb{E}\left[\left(\sum_{b=1}^{|\aj|} \U[a,b\mid\hat{\theta}]\hat{\y}_{j}(b)\right)^2\right] \\
    &\leq \mathbb{E}\left[\|\U[a,:\mid\hat{\theta}]\|_{\infty}^2 \left(\sum_{b=1}^{|\aj|} \hat{\y}_{j}(b)\right)^2\right] \\
    &\leq \mathbb{E}\left[\|\U[a,:\mid\hat{\theta}]\|_{\infty}^2\right] \leq \sigma^2+1
\end{aligned}
\end{equation*}
where the first inequality follows from the non-negativity of both $\U[a,b\mid\hat{\theta}]$ and $\hat{\y}_{j}(b)$, the last inequality is from E.q.~\eqref{equ:Usquare}. Hence, the expected inner product $\mathbb{E}\left[\langle\hatl_{j}^{\,2},\x_{j}\rangle\right]$ is bounded as
\begin{equation}\label{equ:proof-uregret-lnsquare}
\begin{aligned}
    \mathbb{E}\left[\langle\hatl_{j}^{\,2},\x_{j}\rangle\right] &\leq \mathbb{E}\left[ \sum_{a=1}^{|\ar|}\x_{j}(a) \times (\sigma^2+1)\right] =\sigma^2+1.
\end{aligned}
\end{equation}
By substituting E.q.~\eqref{equ:proof-uregret-lnsquare} into E.q.~\ref{equ:proof-uregret}, the universal regret is bounded as
\begin{equation*}
\begin{aligned}
\text{\normalsize U-Regret}_N\text{\normalsize$(\mathcal{A})$}
&\leq \sum_{i=1}^{|\mH|}\sum_{j\in\mathcal{N}_i} \left[ \eta(\sigma^2+1) + \frac{\log{|\ar|}}{\eta} \right] \\
&=\sum_{i=1}^{|\mH|} \left[ \eta(\sigma^2+1)|\mathcal{N}_i| + \frac{\log{|\ar|}}{\eta} \right] \\
&= \eta(\sigma^2+1)N + \frac{|\mH|\log{|\ar|}}{\eta} \\
&\leq 2\sqrt{(\sigma^2+1)N|\mH|\log{|\mathcal{A}_R|}}
\end{aligned}
\end{equation*}
where the second equality follows from the identity $\sum_{i=1}^{|\mH|} |\mathcal{N}_i| = N$, and the last inequality is obtained by selecting the learning rate as $\eta=\sqrt{|\mH|\log{|\ar|}/(\sigma^2+1)N}$.
\end{proof}